\documentclass[%
reprint, onecolumn,
superscriptaddress,
 amsmath,amssymb,
 aps, prx,
longbibliography,
floatfix,
]{revtex4-2}

\usepackage{amsmath, amssymb, amsthm}
\usepackage[utf8]{inputenc}
\usepackage{graphicx}
\usepackage{braket}
\usepackage{bm,bbm}
\usepackage{enumerate}

\newcommand{\abs}[1]{\left\lvert #1 \right\rvert}
\newcommand{\ceil}[1]{\left \lceil #1 \right \rceil}

\def\ketbra#1#2{{\vert#1\rangle\!\langle#2\vert}}

\newcommand{\id}{\mathbbm{1}}

\usepackage{newtxtext}
\usepackage{color}
\usepackage{xcolor}
\usepackage{url}
\usepackage[colorlinks=true, allcolors = blue
]{hyperref}

\usepackage{letltxmacro}
\LetLtxMacro{\oldtextsc}{\textsc}
\renewcommand{\textsc}[1]{\oldtextsc{\scalefont{1.1}#1}}

\theoremstyle{plain}
\newtheorem{theorem}{Theorem}
\newtheorem{lemma}{Lemma}
\newtheorem{corollary}[theorem]{Corollary}

\theoremstyle{definition}

\newcommand{\chemUofT}{\affiliation{%
    Chemical Physics Theory Group, Department of Chemistry, University of Toronto, Toronto, Ontario, Canada}
    }
\newcommand{\csUofT}{\affiliation{%
    Department of Computer Science, University of Toronto, Toronto, Ontario, Canada}
    }
\newcommand{\vecInst}{\affiliation{%
    Vector Institute for Artificial Intelligence, Toronto, Ontario, Canada}
    }
\newcommand{\chemEngUofT}{\affiliation{%
    Department of Chemical Engineering \& Applied Chemistry, University of Toronto, Toronto, Ontario, Canada}
    }
\newcommand{\matSciUofT}{\affiliation{%
    Department of Materials Science \& Engineering, University of Toronto, Toronto, Ontario, Canada}
    }
\newcommand{\CIFAR}{\affiliation{%
    Lebovic Fellow, Canadian Institute for Advanced Research, Toronto, Ontario, Canada}
    }

\begin{document}
\title{Efficient Quantum Algorithm for All Quantum Wavelet Transforms}
\author{Mohsen Bagherimehrab}
\email{mohsen.bagherimehrab@gmail.com}
\chemUofT\csUofT 
\author{Al\'an Aspuru-Guzik}
\chemUofT\csUofT\vecInst\chemEngUofT\matSciUofT\CIFAR

\date{\today}

\date{\today}

\begin{abstract}
Wavelet transforms are widely used in various fields of science and engineering as a mathematical tool with features that reveal information ignored by the Fourier transform. Unlike the Fourier transform, which is unique, a wavelet transform is specified by a sequence of numbers associated with the type of wavelet used and an order parameter specifying the length of the sequence. 
While the quantum Fourier transform, a quantum analog of the classical Fourier transform, has been pivotal in quantum computing, prior works on quantum wavelet transforms~(QWTs) were limited to the second and fourth order of a particular wavelet, the Daubechies wavelet.
Here we develop a simple yet efficient quantum algorithm for executing any wavelet transform on a quantum computer.
Our approach is to decompose the kernel matrix of a wavelet transform as a linear combination of unitaries (LCU) that are compilable by easy-to-implement modular quantum arithmetic operations and use the LCU technique to construct a probabilistic procedure to implement a QWT with a \textit{known} success probability.
We then use properties of wavelets to make this approach deterministic by a few executions of the  amplitude amplification strategy.
We extend our approach to a multilevel wavelet transform and a generalized version, the packet wavelet transform, establishing computational complexities in terms of three parameters:
the wavelet order~$M$, the dimension~$N$ of the transformation matrix, and the transformation level~$d$.
We show the cost is logarithmic in~$N$, linear in $d$ and superlinear in~$M$.
Moreover, we show the cost is independent of $M$ for practical applications.
Our proposed quantum wavelet transforms could be used in quantum computing algorithms in a similar manner to their well-established counterpart, the quantum Fourier transform.
\end{abstract}

\maketitle

\section{Introduction}
As a solid alternative to the Fourier transform, wavelet transforms are a relatively new mathematical tool with diverse utility that has generated much interest in various fields of science and engineering over the past four decades. Although wavelet-like functions have existed for over a century, a prominent example is what is now known as the Haar wavelet. The interest is due to the attractive features of wavelets~\cite{Mal09,Dau92,Bey92,mbm}. Such functions are differentiable, up to a particular order, and are local in both the real and dual spaces. They provide an exact representation for polynomials up to a certain order, and a simple yet optimal preconditioner for a large class of differential operators. Crucially, wavelets provide structured and sparse representations for vectors, functions, or operators, enabling data compression and constructing faster algorithms. These appealing features of wavelets and their associated transforms make them advantageous for numerous applications in classical computing over their established counterpart, the Fourier transform.

With the wavelet transforms' diverse utility and extensive use in classical computing, a natural expectation is that a quantum analog of such transforms will find applications in quantum computing, especially for developing faster quantum algorithms and quantum data compression. Wavelets have already been used in quantum physics and computation~\cite{BP13,BRS+15,BSB+22,HTC+22,ES16,Poly20,GSM+22,BNW+23}. However, prior works on developing a quantum analog for wavelet transforms are limited to a few representative cases~\cite{FW99,hoy97,LFX+18,LFX+19,LLX23}. In contrast, the quantum Fourier transform, a quantum analog of the classical Fourier transform, has been extensively used in quantum computing as a critical subroutine for many quantum algorithms.

Unlike the Fourier transform, a wavelet transform is not unique and is specified by the type of wavelet used and an order parameter. In particular, a wavelet transform is defined by a sequence of numbers, known as the filter coefficients, associated with the type of wavelet used and an even number known as the order of the wavelet that specifies the length of the sequence. Given the sequence, a unitary matrix known as the kernel matrix of the wavelet transform is constructed, the application of which on a vector yields the single-level wavelet transform of the vector. Such a transform partitions the vector into two components: a low-frequency or average component and high-frequency or difference component (see FIG.~\ref{fig:visQWTs}).
To expose the multi-scale~structure of the vector, or a function for that matter, the wavelet transform is recursively applied to the low-frequency component, yielding the multi-level wavelet transform of the vector. The wavelet packet transform is a generalization of the multi-level wavelet transform, in which the wavelet transform is recursively applied to both the low- and high-frequency components.
We refer to a quantum analog of the (single-) multi-level and packet wavelet transforms as the (single-) multi-level and packet QWTs, respectively.

This paper proposes and analyzes a conceptually simple and computationally efficient quantum algorithm for executing single-level, multi-level, and packet QWTs associated with any wavelet and any order on a quantum computer.
Our approach is based on decomposing a unitary associated with a wavelet transform in terms of a linear combination of a finite number of simple-to-implement unitaries and using the linear combination of unitaries (LCU) technique~\cite{CW12} to implement the original unitary. Specifically, we decompose the kernel matrix of the wavelet transform, associated with a wavelet of order~$M$, as a linear combination of~$M$ simple-to-implement unitaries and, by the LCU technique, construct a probabilistic procedure for implementing the single-level QWT.
The success probability of this approach is a known constant by properties of the wavelet filters. We use this known success probability to make the implementation deterministic using a single ancilla qubit and a few rounds of amplitude amplification. 

Having an implementation for the single-level QWT and recursive formulae describing the multi-level and packet wavelet transforms based on single-level transforms, we construct quantum algorithms for multi-level and packet QWTs.
We establish the computational complexity of these transformations in terms of three parameters:
the wavelet order~$M$,
the dimension of the wavelet-transform matrix~$N$,
and the level of the wavelet transform~$d$.
Without loss of generality, we assume that our main parameter of interest~$N$ is a power of two, as $N=2^n$, and report the computational costs with respect to $n$, the number of qubits that the wavelet transforms act on.

We summarize our main results on computational costs of the described transformations in the following three theorems.
We establish these theorems in subsequent sections after providing a detailed description of our algorithms. 

\begin{theorem}[Single-level QWT with logarithmic gate cost]
\label{theorem:1QWT}
    A single-level QWT on $n$ qubits, associated with a wavelet of order~$M$, can be implemented using
    $\ceil{\log_2M}+1$
    ancilla qubits and $\mathcal{O}(n)+\mathcal{O}(M^{3/2})$ \textup{\textsc{toffoli}} and elementary one- and two-qubit gates.
\end{theorem}

\begin{theorem}[Multi-level QWT with multiplicative gate cost]
\label{theorem:dQWT}
    A d-level QWT on $n$ qubits
    can be achieved using $\ceil{\log_2M}+2$ ancilla qubits and $\mathcal{O}(dn)+\mathcal{O}(dM^{3/2})$ \textup{\textsc{toffoli}} and elementary one- and two-qubit gates.
\end{theorem}

\begin{theorem}[Packet QWT]
\label{theorem:pQWT}
    A d-level packet QWT on $n$ qubits
    can be achieved using $\ceil{\log_2M}+1$ ancilla qubits and 
    $\mathcal{O}(dn-d^2/2)+\mathcal{O}(dM^{3/2})$ \textup{\textsc{toffoli}} and elementary one- and two-qubit gates.
\end{theorem}

We remark that the number of levels for the multi-level or packet QWTs is upper bounded by~$n$, i.e., $d\leq n$.
Hence, as a corollary of Theorems~\ref{theorem:dQWT} and~\ref{theorem:pQWT} , the gate cost for these transformations is at most quadratic in $n=\log_2N$.
We show that the gate costs reported in the above theorems are independent of $M$ for practical applications and only a few number of ancilla qubits suffice to implement the multi-level and packet QWTs.
We discuss allowable range for the order parameter~$M$ versus the values used in practical applications in the discussion section.

The rest of this paper proceeds as follows. 
We begin by describing the notation we use throughout the paper.
Then we detail our approach for implementing a single-level QWT by simple modular arithmetic operations in~\S\ref{sec:1QWT}.
We describe the multi-level and packet QWT in~\S\ref{sec:dQWT}, followed by detailed complexity analysis for our algorithms in~\S\ref{sec:complexity}.
Finally, we discuss our results and conclude in~\S\ref{sec:conclusion}.

\textbf{Notation}:
We refer to $A\in \mathbb{C}^{2^n\times 2^n}$ as $n$-qubit matrix and denote the $n$-qubit identity by $\mathbbm{1}_n$.
Throughout the paper, we use the symbol~$M$ for the wavelet order and $m=\ceil{\log_2M}$.
The wavelet order is an even positive number as $M=2\mathcal{K}$ with $\mathcal{K}$ a positive integer called the wavelet index;
the symbol $\mathcal{K}$ is used for $M/2$.
We use zero indexing for iterable mathematical objects such as vectors and matrices.
Qubits of an $n$-qubit register is ordered from right to left, i.e., the rightmost~(leftmost) qubit in $\ket{q_{n-1},\ldots,q_1,q_0}$ representing the state of an $n$-qubit register that encodes the binary representation of an integer $q$ is the first~(last) qubit.
The first and last qubits are also referred to as the least-significant bit~(LSB) and the most-significant bit~(MSB).
Qubits in a quantum circuit are ordered from bottom to top: the bottom qubit is the LSB and the top qubit is the MSB.

\section{Single-level QWT}
\label{sec:1QWT}
This section describes our algorithm for executing a single-level wavelet transform on a quantum computer.
Such a transformation is specified by a kernel matrix.
We describe this matrix in \S\ref{subsec:kernelMtx} and decompose it as a linear combination of a finite number of unitaries.
The decomposition enables a prepare-select-unprepare-style procedure for probabilistic implementation of the desired transformation that we cover in \S\ref{subsec:LCU}.
In \S\ref{subsec:success}, we describe how purposefully reducing the success probability yields a perfect amplitude amplification.
Finally, in \S\ref{subsec:select} and \S\ref{subsec:prep}, we provide a compilation for the select and prepare operations based on simple-to-implement modular arithmetic operations.
\begin{figure*}
\centering
    \includegraphics[width=\linewidth]{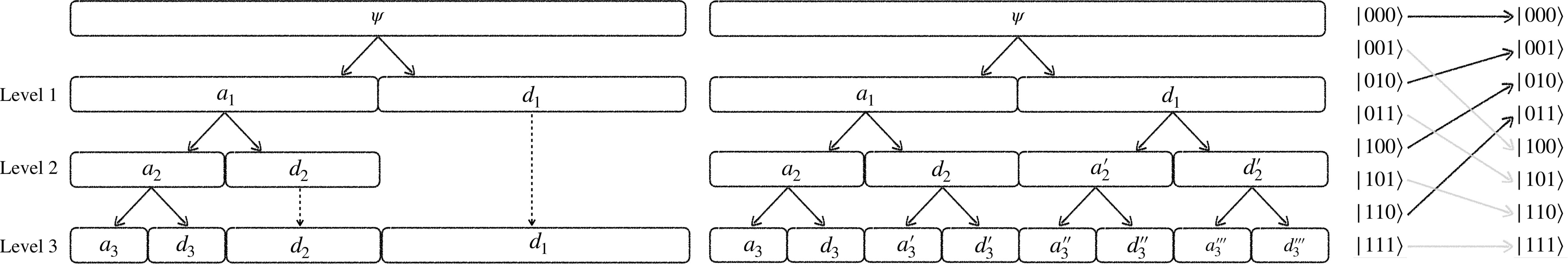}
    \caption{
    Visualization for (Left)~multi-level wavelet transform,
    (Middle) the packet wavelet transform, and
    (Right)~action of quantum perfect shuffle transform in Eq.~\eqref{eq:shuffle} on three-qubit basis states.
    Three-level wavelet transforms are shown for simplicity.
    In the first level, the size-$N$ vector $\psi$ is partitioned into two size-$N/2$ vectors:
    an average vector~$a_1=H\psi$ and a difference vector~$d_1=G\psi$ with~$H$ and $G$ defined in Eq.~\eqref{eq:HandG}.
    For the quantum (packet) wavelet transform, the components of~$\psi$ are amplitudes of a quantum state.
    The wavelet transform is recursively applied to the average vector in the multi-level wavelet transform. In contrast, the packet transform applies the wavelet transform to both the average and difference vectors.
    }
    \label{fig:visQWTs}
\end{figure*}
\subsection{The wavelet kernel matrix as a linear combination of unitaries}
\label{subsec:kernelMtx}

We begin this subsection by briefly describing the kernel matrix associated with a wavelet transform.
We refer to~\cite[Chap.~2.1]{mbm} for a review of wavelet formalism and how this matrix is constructed.
The kernel matrix $W$ of a wavelet transform is specified by the wavelet filter coefficients:
a sequence of numbers $(h_0,h_1,\ldots,h_M)$ that depend on the type of wavelet and satisfy 
\begin{equation}
\label{eq:conds}
    \sum_{\ell=0}^{M-1} h_\ell = \sqrt{2},
    \quad \sum_{\ell=0}^{M-1} h_\ell^2 = 1,
\end{equation}
where the even number $M$ is the wavelet order.
Specifically, the $2^n\times 2^n$ kernel matrix~$W$ is comprised of $2^{n-1}\times 2^n$ matrices $H$ and $G$ as
\begin{equation}
\label{eq:HandG}
   W =\begin{bmatrix} H\\G \end{bmatrix},
   \quad
   H_{ij}=
   h_{j-2i\, (\bmod 2^n)},
   \quad
   G_{ij}=
   g_{j-2i\, (\bmod 2^n)},
   \quad
   g_\ell = (-)^\ell h_{M-1-\ell}
   ;
\end{equation}
an example of the kernel matrix $W$ for a forth-order wavelet $(M=4)$ is as follows

\begin{equation}
\label{eq:WTM}
   W = 
   {\small
   \begin{bmatrix}
      h_0  &   h_1  &   h_2  &   h_3  &    0   &    0   & \cdots & 0 & 0 & 0 & 0 \\
       0   &    0   &   h_0  &   h_1  &   h_2  &   h_3  & \cdots & 0 & 0 & 0 & 0\\
       0   &    0   &    0   &    0   &   h_0  &   h_1  & \cdots & 0 & 0 & 0 & 0\\
    \vdots & \vdots & \vdots & \vdots & \vdots & \vdots & \ddots & \vdots & \vdots & \vdots & \vdots\\
      0    &   0    &    0   &    0   &    0   &    0   & \cdots & h_0 & h_1 & h_2 & h_3\\
      h_2  &   h_3  &    0   &    0   &    0   &    0   & \cdots & 0 & 0 & h_0 & h_1\\
      h_3  &  -h_2  &   h_1  &  -h_0  &    0   &    0   & \cdots & 0 & 0 & 0 & 0\\
       0   &    0   &   h_3  &  -h_2  &    h_1 &   -h_0 & \cdots & 0 & 0 & 0 & 0\\
       0   &    0   &   0    &  0     &    h_3 &   -h_2 & \cdots & 0 & 0 & 0 & 0\\
    \vdots & \vdots & \vdots & \vdots & \vdots & \vdots & \ddots &\vdots & \vdots & \vdots & \vdots\\
      0    &   0    &    0   &    0   &    0   &    0   & \cdots & h_3 & -h_2 & h_1 & -h_0\\
      h_1  &  -h_0  &   0    &    0   &    0   &   0    & \cdots & 0 & 0 & h_3  & -h_2\\
   \end{bmatrix}
   },\;\,
   U = 
   {\small   
   \begin{bmatrix}
      h_0  &   h_1  &   h_2 &   h_3 &    0   &    0     &\cdots & 0     & 0     & 0     & 0     \\
       0   &    0   &   h_0 &   h_1 &   h_2  &   h_3    &\cdots & 0     & 0     & 0     & 0     \\
       0   &    0   &    0  &    0  &   h_0  &   h_1    &\cdots & 0     & 0     & 0     & 0     \\
    \vdots & \vdots & \vdots& \vdots& \vdots & \vdots   &\ddots &\vdots &\vdots &\vdots &\vdots \\
       0   &    0   &    0  &    0  &    0   &    0     &\cdots & h_0   & h_1   & h_2   & h_3   \\
      h_2  &   h_3  &    0  &    0  &    0   &    0     &\cdots & 0     & 0     & h_0   & h_1   \\
      h_1  &  -h_0  &    0  &    0  & \cdots &    0     & 0     & 0     & 0     & h_3   & -h_2  \\
      h_3  &  -h_2  &   h_1 &  -h_0 & \cdots &    0     & 0     & 0     & 0     & 0     & 0     \\
    \vdots &\vdots  &\vdots &\vdots &\ddots  & \vdots   &\vdots &\vdots &\vdots &\vdots &\vdots \\
        0  &  0     &   0   &  0    & \cdots & h_1      & -h_0  & 0     & 0     & 0     & 0     \\
        0  &  0     &   0   &  0    & \cdots & h_3      & -h_2  & h_1   & -h_0  & 0     & 0     \\
        0  &  0     &   0   &  0    & \cdots & 0        &  0    & h_3   & -h_2  & h_1   & -h_0  \\
   \end{bmatrix}
   }
\end{equation}
The unitary matrix $U$ here is a modification of the unitary $W$ that we use for decomposing $W$ as a linear combination of unitaries.
To this end, let us first define the circular downshift and upshift permutation operations as
\begin{equation}
\label{eq:shift}
    S_n^{\downarrow} :=
    \sum_{j=0}^{2^n-1} \ketbra{j+1\bmod 2^n}{j}
    =
    \begin{bmatrix}
          &  &          &  &1\\
        1 &  &          &  & \\
          &1 &          &  & \\
          &  &  \ddots  &  & \\
          &  &          &1 &
    \end{bmatrix},
    \quad
    S_n^{\uparrow} := 
    \sum_{j=0}^{2^n-1} \ketbra{j-1\bmod 2^n}{j}
    =
    \begin{bmatrix}
          &1 &  &       & \\
          &  &1 &       & \\
          &  &  &\ddots & \\
          &  &  &       &1\\
        1 &  &  &       &
    \end{bmatrix},
\end{equation}
where the matrix size is $2^n\times 2^n$.
Note that these operations are inverse of each other and their action on $n$-qubit basis state $\ket{j}$ is
\begin{equation}
    S^\downarrow_n\ket{j} = \ket{j+1 \bmod 2^n},
    \quad
    S^\uparrow_n\ket{j} = \ket{j-1 \bmod 2^n}.
\end{equation}
Upon acting on a vector with $2^n$ components, $S_n^\downarrow/S_n^\uparrow$ shifts the vector's components one place downward/upward with wraparound.
Similarly, when acting on a matrix with $2^n$ rows from the left side, $S_n^\downarrow/S_n^\uparrow$ shifts the rows of the matrix one place downward/upward with wraparound.

To construct an LCU decomposition for the $n$-qubit unitary $W$, the kernel matrix associated with a wavelet of order $M=2\mathcal{K}$, first we transform it into another unitary $U$ by $\mathcal{K}-1$ downshift permutations of the rows in the lower half of $W$.
Specifically, we transform $W$ as
\begin{equation}
   W =\begin{bmatrix} H\\G \end{bmatrix} \to U =\begin{bmatrix} H\\G' \end{bmatrix},
   \quad
   G':= (S_{n-1}^{\downarrow})^{\mathcal{K}-1}G,
\end{equation}
where $G'$ is obtained by $\mathcal{K}-1$ downshift permutations of the rows of $G$ and its elements are
\begin{equation}
    G'_{i,j} = (-1)^{2i+2-j}h_{2i+2-j}.
\end{equation}
Let us now represent $\mathcal{K}-1$ upnshift permutations on $n$ qubits by $\textsc{ushift}_n$ with the action
\begin{equation}
\label{eq:ushift}
    \textsc{ushift}_n \ket{j}
    := \ket{j+\mathcal{K}-1\bmod 2^n}
\end{equation}
on $n$-qubit basis state $\ket{j}$.
Then we have
\begin{equation}
\label{eq:WU}
    W= (\ketbra{0}{0}\otimes\id_{n-1}+\ketbra{1}{1}\otimes \textsc{ushift}_{n-1}) U
    := \Lambda_1(\textsc{ushift}_{n-1}) U,
\end{equation}
i.e., $W$ is obtained by $\mathcal{K}-1$ upshift permutations of the rows in the lower half of $U$.

We now decompose the unitary $U$ as a linear combination of $M$ unitaries as
\begin{equation}
\label{eq:LCU}
    U = \sum_{\ell=0}^{M-1} h_\ell U_\ell,
    \quad
    U_\ell :=
    \begin{cases}
        P_\ell
        \quad\quad\quad\quad\quad\quad
        \text{if $\ell$ is odd},\\
        (Z\otimes\id_{n-1})P_\ell
        \quad
        \text{if $\ell$ is even},
    \end{cases}
\end{equation}
where $Z$ is the Pauli-$Z$ operator and 
the unitary 
\begin{equation}
\label{eq:P_ell}
    P_\ell :=
    \sum_{j=0}^{N/2-1}
    \ketbra{j}{2j+\ell \bmod\! N}
    +\ketbra{N/2+j}{2j+1-\ell \bmod\! N}
\end{equation}
is a permutation matrix that is obtained from $U$ as follows:
all entries of $U$ with value $\pm h_\ell$ are replaced with $1$ and all other nonzero entries are replaced with $0$.
Because $W$ is unitarily equivalent to $U$ by Eq.~\eqref{eq:WU}, the LCU decomposition in Eq.~\eqref{eq:LCU} provides a similar LCU decomposition for $W$.  

\subsection{Probabilistic implementation for the single-level QWT}
\label{subsec:LCU}

The decomposition in Eq.~\eqref{eq:LCU} enables a prepare-select-unprepare-style method~\cite{CW12} for probabilistic implementation of $U$.
To this end, let
\begin{equation}
\label{eq:prep}
    \textsc{prep} \ket{0^m}
    := \frac{1}{\sqrt{h}} \sum_{\ell=0}^{M-1} \sqrt{\abs{h_\ell}} \ket{\ell},
    \quad
    \textsc{unprep}^\dagger \ket{0^m}
    := \frac{1}{\sqrt{h}} \sum_{\ell=0}^{M-1} \text{sign}(h_\ell)\sqrt{\abs{h_\ell}} \ket{\ell},
    \quad
    h:=\sum_{\ell=0}^{M-1} \abs{h_\ell},
\end{equation}
where $m=\ceil{\log_2M}$ is the number of ancilla qubits, and let \textsc{select} be an operation such that
\begin{equation}
\label{eq:select}
    \textsc{select} \ket{\ell}\ket{j}
    := \ket{\ell} U_\ell \ket{j}
\end{equation}
with $U_\ell$ defined in Eq.~\eqref{eq:LCU}. 
Then for any $n$-qubit state we have
\begin{equation}
\label{eq:prep_select_prep}
    (\textsc{unprep}\otimes\id_n)\,
    \textsc{select}\,
    (\textsc{prep}\otimes\id_n)
    \ket{0^m}\ket{\psi}\\
    = \frac{1}{h}\ket{0^m} U\ket{\psi}
    +\sqrt{1-\frac{1}{h^2}}\ket{\perp},
\end{equation}
where $\ket{\perp}$ is an ($m+n$)-qubit state such that $(\bra{0^m}\otimes\id_n)\ket{\perp}=0$.
This equation follows as
\begin{align}
    \ket{0^m}\ket{\psi}
    &\xrightarrow{\;\;\textsc{prep}\,\otimes \id_n\;\;}\frac{1}{\sqrt{h}}\sum_\ell \sqrt{\abs{h_\ell}} \ket{\ell} \ket{\psi}\\
    &\xrightarrow{\quad\textsc{select}\quad}
    \frac{1}{\sqrt{h}}\sum_\ell \sqrt{\abs{h_\ell}} \ket{\ell} U_\ell\ket{\psi}\\
    &\xrightarrow{\textsc{unprep}\,\otimes \id_n}
    \frac{1}{h}\ket{0^m} U\ket{\psi}
    +\sqrt{1-\frac{1}{h^2}}\ket{\perp},
\end{align}
where the last line follows by projecting the ancilla qubits to $\ket{0^m}$ state, i.e.,
\begin{align}
    (\bra{0^m}\otimes \id_n)(\textsc{unprep}\otimes\id_n)
    \frac{1}{\sqrt{h}}\sum_\ell \sqrt{\abs{h_\ell}} \ket{\ell} U_\ell\ket{\psi}
    &=
    \frac{1}{h} \sum_{\ell'\ell}
    \text{sign}(h_{\ell'}) \sqrt{\abs{h_{\ell'} h_\ell}} \braket{\ell'|\ell} U_\ell \ket{\psi}\\
    &= \frac{1}{h} \sum_\ell
    h_\ell U_\ell \ket{\psi}=\frac{1}{h}U\ket{\psi}.
\end{align}
Equation~\eqref{eq:prep_select_prep} yields a probabilistic implementation for $U$.
Because $U$ and $W$ are unitarily equivalent, by Eq.~\eqref{eq:WU}, we also have a probabilistic implementation for $W$ with the same success probability.
In particular, let us define a probabilistic QWT as
\begin{equation}
\label{eq:pqwt}
    \textsc{pqwt}
    := (\id_m\otimes\Lambda_1(\textsc{ushift}_{n-1})) (\textsc{unprep}\otimes\id_n)\,
    \textsc{select}\,
    (\textsc{prep}\otimes\id_n),
\end{equation}
then we have
\begin{equation}
\label{eq:pqwt_action}
    \textsc{pqwt} \ket{0^m}\ket{\psi}
    = \sin(\alpha)\ket{0^m} W\ket{\psi}
    +\cos(\alpha)\ket{\perp'},
    \quad
    \sin(\alpha):=1/h,
\end{equation}
with the $(m+n)$-qubit state $\ket{\perp'}:= \id_m\otimes\Lambda_1(\textsc{ushift}_{n-1}) \ket{\perp}$ and the $\ket{1}$-controlled unitary $\Lambda_1(\textsc{ushift})$ defined in Eq.~\eqref{eq:WU}.
The success amplitude of this approach is known and its value is $1/h$. As shown in FIG.~\ref{fig:succAmp}(Left), the success amplitude is greater than $1/4$ for a wide range of wavelet order.

We now present an alternative approach for a probabilistic implementation of the single-level QWT. The state-preparation of this approach is simpler and could be preferred in practical applications.
Instead of preparing the state with square-root coefficients by \textsc{prep} in Eq.~\eqref{eq:prep}, in this approach we use the operation \textsc{linprep} defined as
\begin{equation}
\label{eq:linprep}
    \textsc{linprep} \ket{0^m}
    := \sum_{\ell=0}^{M-1} h_\ell \ket{\ell},
\end{equation}
which prepares the state with linear coefficients.
For any $n$-qubit state $\ket{\psi}$ we then have
\begin{equation}
\label{eq:linprep_select_prep}
    (H^{\otimes m}\otimes\id_n)\,
    \textsc{select}\,
    (\textsc{linprep}\otimes\id_n)
    \ket{0^m}\ket{\psi}\\
    = \sin(\alpha)\ket{0^m} U\ket{\psi}
    +\cos(\alpha)\ket{\perp},
    \quad
    \sin(\alpha) := 1/\sqrt{M},
\end{equation}
where $\ket{\perp}$ and \textsc{select} are as those in Eq.~\eqref{eq:prep_select_prep}.
This equation follows as
\begin{equation}
      \ket{0^m}\ket{\psi}
    \xrightarrow{\textsc{linprep}\,\otimes \id_n}
    \sum_\ell h_\ell \ket{\ell} \ket{\psi}
    \xrightarrow{\;\textsc{select}\;}
    \sum_\ell h_\ell \ket{\ell} U_\ell\ket{\psi}
    \xrightarrow{\,H^{\otimes m}\otimes \id_n\,}
    \sin(\alpha)\ket{0^m} U\ket{\psi} + \cos(\alpha)\ket{\perp},  
\end{equation}
where the last step is obtained by projecting the ancilla qubits to $\ket{0^m}$ state, i.e.,
\begin{equation}
    (\bra{0^m}\otimes \id_n)(H^{\otimes m} \otimes\id_n)
    \sum_\ell h_\ell \ket{\ell} U_\ell\ket{\psi}
    =
    \frac{1}{\sqrt{M}} \sum_{\ell'\ell}
    h_\ell \braket{\ell'|\ell} U_\ell \ket{\psi}
    = \frac{1}{\sqrt{M}}\sum_\ell
    h_\ell U_\ell \ket{\psi}=
    \sin(\alpha) U\ket{\psi}.
\end{equation}
The success amplitude of this approach is $1/\sqrt{M}$.
As shown in FIG.~\ref{fig:succAmp}(Right), the magnitude of wavelet coefficients $h_\ell$ with high index~$\ell$ are negligibly small.
Consequently, the success amplitude becomes effectively independent of $M$ for practical~applications.

\begin{figure*}
\centering
    \includegraphics[width=.98\linewidth]{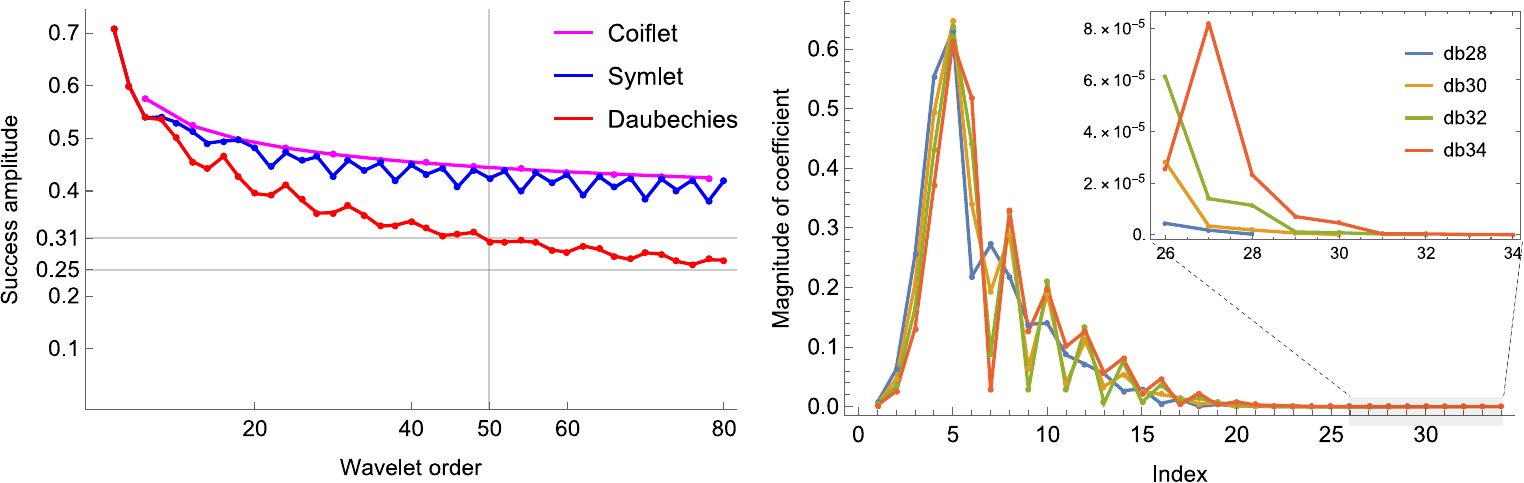}
    \caption{
    (Left)~The success amplitude of the probabilistic implementation in Eq.~\eqref{eq:pqwt_action} for a single-level QWT for commonly used wavelets.
    The success amplitude is known and is greater than $1/4$ for a wide range of wavelet order; it is greater than $0.31$ for the range of wavelet order used in practical~applications.
    For perfect amplitude amplification, the former~(latter) value needs three~(two) rounds of amplitude amplification.
    (Right)~The magnitude of wavelet coefficient $\abs{h_\ell}$ as a function of the index $\ell$ for the Debauchees wavelet with order $M=28,30,32,34$.
    The zoomed-in part shows the wavelet coefficients with higher indexes are negligibly small, and the number of small coefficients increases by increasing the wavelet order.
    The magnitude of wavelet coefficients for other wavelets has a similar pattern. 
    }
    \label{fig:succAmp}
\end{figure*}
\subsection{Reduction of success amplitude for perfect amplitude amplification}
\label{subsec:success}

The success amplitude of the described probabilistic approaches for implementing the single-level QWT is a known constant value.
For perfect amplitude amplification, we purposefully reduce the success amplitude using one extra ancilla qubit.
This end is achieved by applying a rotation gate on the extra qubit initialized in $\ket{0}$.
A few rounds of amplitude amplification then yields the success state with unit probability.

The success amplitude of the probabilistic implementation \textsc{pqwt} in Eq.~\eqref{eq:pqwt_action} is $\sin\alpha$.
This amplitude is known and has a value greater than $1/4$ as discussed in \S\ref{subsec:LCU}.
Let $\theta<\alpha$ be the angle defined in the equation below and let $R(\theta)$ be the rotation gate defined as $R(\theta)\ket{0}:=\cos\theta\ket{0}
+\sin\theta\ket{1}$, then
by Eq.~\eqref{eq:pqwt_action} and for any $n$-qubit state $\ket{\psi}$ we have
\begin{equation}
    (R(\theta)\otimes\textsc{pqwt}) \ket{0}\ket{0^m}\ket{\psi}
    =\sin(\pi/14)\ket{0^{m+1}}W\ket{\psi}
    +\cos(\pi/14)\ket{\perp''},
    \quad \cos\theta:=\sin(\pi/14)/\sin\alpha,
\end{equation}
where $\ket{\perp''}$ is an $(m+n+1)$-qubit state that satisfies
$(\bra{0^{m+1}}\otimes\id_n)\ket{\perp''}=0$.
The success amplitude is now $\sin(\pi/14)$, enabling perfect amplitude amplification.
Indeed, by only three rounds of amplitude amplification,
$W$ is applied on~$\ket{\psi}$ and all $m+1$ ancilla qubits end up in the all-zero state.

We remark that the success amplitude is grater than~$0.31$ for the range of wavelet order used
in practical applications; see FIG.~\ref{fig:succAmp}(Left).
In this case, we reduce the success amplitude to $\sin(\pi/10)<0.31$ by setting $\cos\theta=\sin(\pi/10)/\sin\alpha$ and achieve the perfect amplitude amplification by only two rounds of amplitude amplification.

We use the oblivious amplitude amplification because the input state $\ket{\psi}$ is unknown.
To this end, let $R_n = 2\ketbra{0^n}{0^n}-\id_n$ be the $n$-qubit reflection operator with respect to the $n$-qubit zero state $\ket{0^n}$ and let
\begin{equation}
\label{eq:AAoprator}
    \mathcal{A}:=
    -(R(\theta)\otimes\textsc{pqwt})
    (R_{m+1}\otimes \id_n)
    (R(\theta)\otimes\textsc{pqwt})^\dagger
    (R_{m+1}\otimes \id_n),
\end{equation}
be the amplitude amplification operator.
Then the following holds~\cite[Lemma~2.2]{kothari}
\begin{equation}
\label{eq:OAA}
    \mathcal{A}^t (R(\theta)\otimes\textsc{pqwt})
    \ket{0}\ket{0^m}\ket{\psi}
    = \sin((2t+1)\pi/14)
    \ket{0^{m+1}}W\ket{\psi}
    +\cos((2t+1)\pi/14) \ket{\perp''}.
\end{equation}
Therefore, the unit success probability is achieved by three executions of amplitude amplification~($t=3$).

The success amplitude of the second approach based on Eq.~\eqref{eq:linprep_select_prep} is $\sin\alpha=1/\sqrt{M}$.
In this case, we reduce the success amplitude to $\sin(\pi/2(2t+1))$ by applying the rotation gate $R(\theta)$ on the extra qubit initialized in $\ket{0}$ state, with $t$ and $\theta$ defined as
\begin{equation}
    t:= \ceil{\frac{1}{2}\left(\frac{\pi}{2\alpha}-1\right)},
    \quad
    \cos\theta := \frac{\sin(\pi/2(2t+1))}{\sin\alpha} = \sqrt{M} \sin\left(\frac{\pi}{2}\frac{1}{2t+1}\right).
\end{equation}
Then we achieve the desired state with unit success probability by $t$ rounds of amplitude amplification, i.e.,
$W$ is applied on~$\ket{\psi}$ and all $m+1$ ancilla qubits end up in the all-zero state.

\subsection{Implementing \textsc{select} by modular quantum arithmetic}
\label{subsec:select}

Here we describe our approach for implementing the \textsc{select} operation by simple modular arithmetic operations on a quantum computer.
As per Eq.~\eqref{eq:select}, \textsc{select} applies $U_\ell$ on the second register $\ket{j}$ based on the value of $\ell$ encoded in the first register~$\ket{\ell}$.
If $\ell$ is odd, then $U_\ell = P_\ell$ by Eq.~\eqref{eq:LCU}.
Otherwise, $U_\ell$ is a product of $P_\ell$ and a single Pauli-$Z$ on the first qubit of the second register.
That is to say that $U_\ell$ and $P_\ell$ are equivalent up to a $\ket{0}$-controlled-$Z$ operation; control qubit is the qubit representing the least-significant bit (LSB) of $\ell$ and target qubit is the one representing the most-significant bit (MSB) of $j$.   
Implementing \textsc{select} is therefore achieved by an implementation for $P_\ell$.

The $n$-qubit permutation $P_\ell$ in Eq.~\eqref{eq:P_ell} transforms the $n$-qubit basis state $\ket{j}$ as
\begin{equation}
\label{eq:P_ell_action}
    P_\ell: \ket{j}\mapsto
    \begin{cases}
        \Ket{\frac{j-\ell \bmod\! N}{2}}
            \quad \quad\quad\quad\,
            \text{if $j$ and $\ell$ have same parity,} \vspace{1mm} \\   
        \Ket{\frac{N}{2}+\frac{j+\ell-1 \bmod\! N}{2}}   
            \quad
            \text{if $j$ and $\ell$ have opposite parity.}
    \end{cases}
\end{equation}
This transformation can be implemented by modular quantum addition \textsc{add} and subtraction \textsc{sub} defined as
\begin{equation}
\label{eq:add}
    \textsc{add}\ket{\ell}\ket{j} := \ket{\ell} \ket{j+\ell \bmod\! N},
    \quad
    \textsc{sub}\ket{\ell}\ket{j} := \ket{\ell} \ket{j-\ell \bmod\! N},
\end{equation}
and by the quantum perfect shuffle transformation defined as
\begin{equation}
\label{eq:shuffle}
    \textsc{shuffle}
    \ket{q_{n-1}\ldots q_1q_0}
    := \ket{q_0q_{n-1}\ldots q_1},
\end{equation}
which performs the transformation $\ket{q}\mapsto\ket{q/2}$ if $q$ is an even number and $\ket{q}\mapsto\ket{N/2+(q-1)/2}$ if $q$ is odd (see FIG.~\ref{fig:visQWTs}). For clarity, we remark that here and in the following $\ket{\ell}$ is an $m$-qubit basis state with $m<n$ and $\ket{j}$ is an $n$-qubit basis state.

To implement $P_\ell$ by these operations, we use a single ancilla qubit called parity qubit and define the parity operation \textsc{par} as 
\begin{equation}
\label{eq:par}
    \textsc{par}\ket{0}\ket{\ell}\ket{j}:=
    \begin{cases}
        \ket{0}\ket{\ell}\ket{j}
            \quad \text{if $j$ and $\ell$ have same parity,} \vspace{1mm} \\   
        \ket{1}\ket{\ell}\ket{j}
            \quad \text{if $j$ and $\ell$ have opposite parity,}
    \end{cases}
\end{equation}
which flips the parity qubit based on the parity of $\ell$ and $j$;
parity of a number is $0$ if its even and is $1$ otherwise.
This operation can be implemented using two \textsc{cnot} gates, one controlled on the LSQ of the register encoding $\ell$ and the other controlled on the LSQ of the register encoding $j$.
The target qubit for each \textsc{cnot} is the parity qubit.

Having computed the parity by \textsc{par}, we then apply \textsc{sub} to the last two registers if the parity qubit is $\ket{0}$ and apply \textsc{add} to these registers if the parity is $\ket{1}$, followed by the shuffle operation in Eq.~\eqref{eq:shuffle} on the last register.
By these operations, the state of the parity qubit, the $m$-qubit register encoding $\ell$, and the $n$-qubit register encoding $j$ transform as
\begin{align}
    \textsc{par}\ket{0}\ket{\ell}\ket{j}
    &\xrightarrow{\ketbra{0}{0}\otimes\, \textsc{sub} + \ketbra{1}{1}\otimes\,\textsc{add}}
    \ket{0}\ket{\ell}\ket{j-\ell\bmod\!N}+
    \ket{1}\ket{\ell}\ket{j+\ell\bmod\!N}\\
    &\xrightarrow{\,\quad \id_1\otimes\id_m\otimes\,\textsc{shuffle} \,\quad}
    \ket{0}\ket{\ell}\ket{(j-\ell\bmod\!N)/2}+
    \ket{1}\ket{\ell}\ket{N/2 + (j+\ell-1\bmod\!N)/2},
\end{align}
where $N=2^n$.
We finally erase the parity qubit to achieve an implementation for $P_\ell$.
To this end, we note that the parity qubit is $\ket{1}$ only if the value encoded in the last register is greater than $N/2$; see Eq.~\eqref{eq:P_ell_action}.
Hence a \textsc{cnot} from the qubit representing the MSB of the value encoded in the system register to the parity qubit would erase this qubit.

The quantum circuit in the dotted-line box in FIG.~\ref{fig:1QWT} gives an implementation for the \textsc{select} operation based on the described approach.
The sequence of \textsc{swap} gates in this circuit gives a gate-level implementation for \textsc{shuffle} in Eq.~\eqref{eq:shuffle}

\begin{figure*}
\centering
    \includegraphics[width=.8\linewidth]{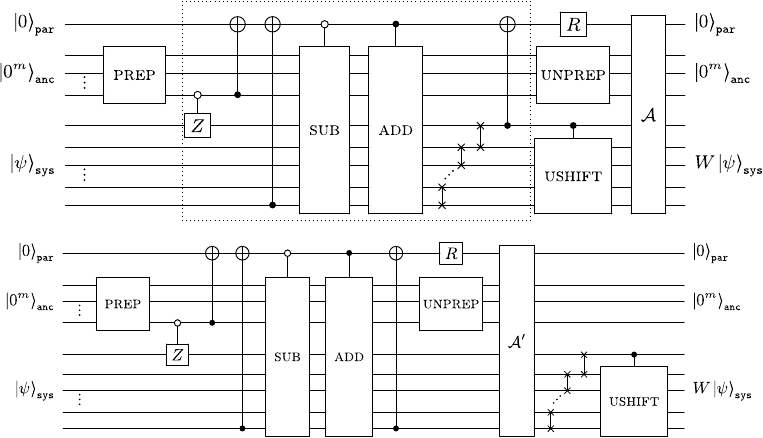}
    \caption{
    Equivalent quantum circuits for executing a single-level QWT comprised of high-level operations.
    Three registers are used:
    the parity register 
    \texttt{par} (one qubit),
    the ancilla register \texttt{anc} ($m$ qubits), and system register \texttt{sys} ($n$ qubits).
    The state of \texttt{sys} register is in a superposition of $\ket{j}$ states for different values of $j$ and the state of \texttt{anc} register, after applying \textsc{prep} with action given in Eq.~\eqref{eq:prep}, is in a superposition of~$\ket{\ell}$ states for different values of~$\ell$.
    The gates inside the dotted-line box implement the \textsc{select} operation in Eq.~\eqref{eq:select} as follows.
    The $\ket{0}$-controlled $Z$ is applied as per Eq.~\eqref{eq:LCU}.
    The first two \textsc{cnot}s compute the parity of $j$ and $\ell$ by their~LSB.
    Then controlled on the parity qubit, we apply \textsc{sub} (parity zero) or \textsc{add} (parity one).
    The sequence of \textsc{swap} gates implement the \textsc{shuffle} operation in Eq.~\eqref{eq:shuffle}.
    The subsequent \textsc{cnot} resets the parity qubit to~$\ket{0}$ because the state of \texttt{par} is filliped to $\ket{1}$ only if $j$ and $\ell$ have opposite parity as per Eq.~\eqref{eq:par}; otherwise it stays $\ket{0}$.
    If \texttt{par} is~$\ket{1}$, the MSB of \texttt{sys} register is in the state $\ket{1}$ as the value encoded in \texttt{sys} is greater than $N/2$ by Eq.~\eqref{eq:P_ell_action}, so the last \textsc{cnot} resets the parity qubit. 
    The \textsc{cnot} has no action if \texttt{par} is $\ket{0}$. This is because the value encoded in the system register is less than~$N/2$ by Eq.~\eqref{eq:P_ell_action} when $j$ and $\ell$ have same parity. Consequently, the MSB of \texttt{sys} is $\ket{0}$, making the last \textsc{cnot} inactive.
    The controlled-\textsc{ushift} operation, with \textsc{ushift} given in Eq.~\eqref{eq:ushift}, maps the implemented unitary by \textsc{select} to the single-level QWT~$W$ as in Eq.~\eqref{eq:WU}.
    The rotation gate~$R$ is used for amplitude amplification $\mathcal{A}$ given in Eq.~\eqref{eq:AAoprator}.
    The bottom circuit follows from the top circuit.
    The amplitude amplification $\mathcal{A}'$ is unitarily equivalent to $\mathcal{A}$.
    }
    \label{fig:1QWT}
\end{figure*}

\subsection{A compilation for state-preparation operations}
\label{subsec:prep}

Here we provide procedures for implementing the \textsc{linprep} and \textsc{prep} operations that prepare states with linear and square-root coefficients, respectively.
We begin with an implementation for \textsc{linprep} in Eq.~\eqref{eq:linprep} using the rotation gate, defined as
\begin{equation}
    R(\theta_\ell):=
    \begin{bmatrix}
        \cos\theta_\ell & \sin\theta_\ell\\
        -\sin\theta_\ell & \cos\theta_\ell
    \end{bmatrix}
\end{equation}
for some known angle $\theta_\ell$, and the increment gate that preforms the map $\ket{\ell}\mapsto\ket{\ell+1}$ for $\ket{\ell}$ an $m$-qubit basis state.
Notice that the increment gate is indeed the downshift permutation $S^\downarrow_m$ defined in Eq.~\eqref{eq:shift} and its inverse is the upshift permutation $S^\uparrow_m$.

The \textsc{linprep} operation prepares a quantum state with amplitudes given by the wavelet filter $\bm{h} =(h_0,\ldots,h_{M-1})^\top$, a column vector of~$M$ real numbers that satisfy Eq.~\eqref{eq:conds}.
By the procedure given in Ref.~\cite{EW18}, the wavelet filter vector $\bm{h}$ of length $M=2\mathcal{K}$ can be achieved by a sequence of $\mathcal{K}$ unitaries $U_\ell$ as $\bm{h}=U_{\mathcal{K}-1}\cdots U_1U_0\, \textup{e}_{\mathcal{K}}$, where $\textup{e}_\ell$ is the $\ell$th column of the $M$-by-$M$ identity matrix and the unitary $U_\ell$ is constructed from rotation gates $R(\theta_\ell)$ as illustrated in FIG.~\ref{fig:binCirc}(a).
As an example, for $M=6$ we have
\begin{equation}
    \begin{bmatrix}
        h_0\\
        h_1\\
        h_2\\
        h_3\\
        h_4\\
        h_5
    \end{bmatrix}
    =
    \begin{bmatrix}
        &c_2 & s_2&    &    \\
        &-s_2& c_2&    &    \\
        &    &    & c_2& s_2\\
        &    &    &-s_2& c_2\\
        &&&    &    & c_2& s_2\\
        &&&    &    &-s_2& c_2
    \end{bmatrix}
    \begin{bmatrix}
        1& & &    &    \\
        &c_1 & s_1&    &    \\
        &-s_1& c_1&    &    \\
        &    &    & c_1& s_1\\
        &    &    &-s_1& c_1\\
        &    &    & &   &1
    \end{bmatrix}
    \begin{bmatrix}
        1 & & & & &\\
          &1& & & &\\
          & & c_0& s_0 & \\
          & &-s_0& c_0 & \\
          & &    &     &1\\
          & &    &     &&1
    \end{bmatrix}
    \begin{bmatrix}
        0\\
        0\\
        0\\
        1\\
        0\\
        0
    \end{bmatrix}
\end{equation}
where $c_\ell:=\cos\theta_\ell$ and $s_\ell:=\sin\theta_\ell$.

Having classically precomputed the rotation angles $(\theta_0,\theta_1,\ldots,\theta_{\mathcal{K}-1})$
by the procedure in Ref.~\cite{EW18}, we construct a quantum circuit for \textsc{linprep} as follows.
Let $m=\ceil{\log_2M}$.
For $M$ that is not a power of $2$, we pad $(2^m-M)/2$ zeros from left and right to the wavelet filter vector $\bm{h}$ to have a vector as $(0,\ldots,0,h_0,\ldots,h_{M-1},0\ldots,0)^\top$.
Then unitaries $U_\ell$ are modified accordingly so that $U_{\mathcal{K}-1}\cdots U_1U_0\, \textup{e}_{2^{m-1}}$ yields the modified wavelet filter vector.
A diagrammatic representation of this approach is shown in FIG.~\ref{fig:binCirc}(b) for $M=6$.
For each $\theta_\ell$ with even $\ell$, first we shift elements of the vector one place to the right, shown in FIG.~\ref{fig:binCirc}(c) by the right arrow, to be able to apply the rotations in parallel on consequent pairs of the vector elements and then shift the vector elements one place to the~left.
Because the rotations are in parallel, we can decompose the associated unitary as a tensor product of an identity and a rotation gate as $\mathbbm{1}_{m-1}\otimes R_\ell$.
Shifting to the right (left) is implemented by the increment gate (inverse of the increment gate) on a quantum computer.
The inverse of the increment gate is applied $(2^m-M)/2$ times at the end to achieve the desired amplitudes as $(h_0,\ldots,h_{M-1},0\ldots,0)^\top$.
The quantum circuit in FIG.~\ref{fig:binCirc}(d) illustrates the case where $M=6$.

\begin{figure*}
\centering
    \includegraphics[width=.98\linewidth]{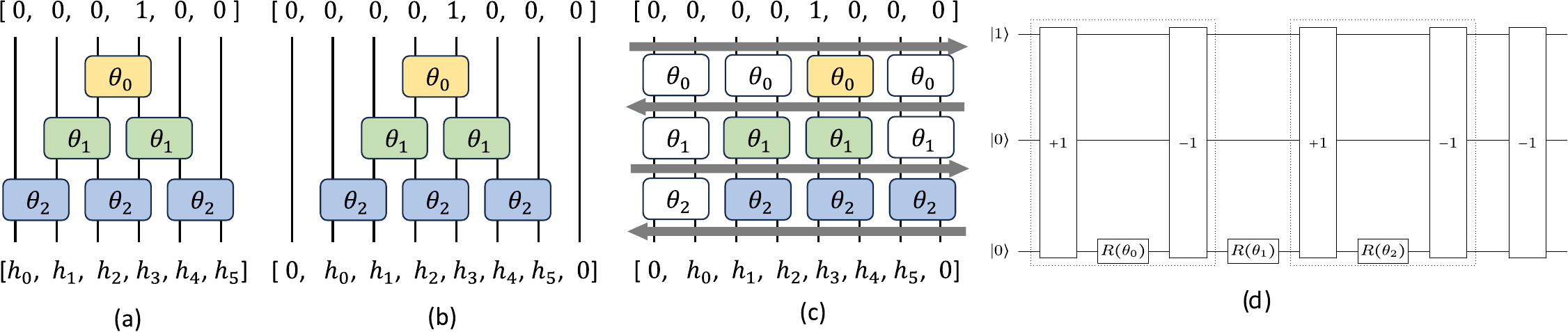}
    \caption{
    (a)~Diagrammatic representation of the procedure producing the wavelet filters for a wavelet of order~$M$ from a particular initial vector by a set of $M/2$ rotations;
    $M=6$ is illustrated here.
    (b)~Zero padding for cases that $M$ is not a power of two.
    (c)~Rotations can be applied in parallel. The right~(left) arrow represents shifting elements of the vector one place to the right~(left). As the initial vector is a particular vector, the rotations represented by white boxes do not affect the vector.
    (d)~Quantum circuit for \textsc{linprep} using rotation gates, the increment gate denoted by $+1$ and its inverse denoted by~$-1$.
    The gate $+1~(-1)$ is applied before (after) each rotation gate $R_\ell$ with even $\ell$, as in dotted boxes.
    }
    \label{fig:binCirc}
\end{figure*}

We now describe an approach for implementing the \textsc{prep} operation in Eq.~\eqref{eq:prep}.
This operation prepares the state with square-root coefficients, i.e., the state
$\ket{\psi}:=\sum_\ell \sqrt{p_\ell}\ket{\ell}$ with $p_\ell:=\abs{h_\ell}/h$.
To prepare this state, first we prepare the uniform superposition state $(1/\sqrt{M})\sum_\ell \ket{\ell}$ and then apply the uniformly controlled rotation~\cite{MVB+04} that performs the map
\begin{equation}
    \label{eq:ucr}
    \ket{0}\ket{\ell} \mapsto
    (\cos(\theta_\ell) \ket{0}+\sin(\theta_\ell)\ket{1})\ket{\ell},
    \quad
    \cos(\theta_\ell) := \sqrt{\abs{h_\ell}/h}.
\end{equation}
The output state after this operation is $\sin\alpha\ket{0}\ket{\psi}+\cos\alpha\ket{\perp}$ with the success amplitude $\sin\alpha:=1/\sqrt{M}$.
As per the discussion in \S\ref{subsec:success}, the state $\ket{\psi}$ is achieved using one extra qubit and $\Theta(\sqrt{M})$ rounds of amplitude amplification.
We remark that the same approach can be used to implement \textsc{unprep} in Eq.~\eqref{eq:prep}.

\section{Multi-level and packet QWT}
\label{sec:dQWT}

We now use our implementation for the single-level QWT as a subroutine and construct quantum algorithms for multi-level and packet QWTs.
To this end, let $W_n^{(d)}$ denote the $d$-level wavelet transform of size $2^n\times 2^n$ and let $P_n^{(d)}$ denote the $d$-level wavelet packet transform of the same size.
Also let $W_n^{(1)}=W_n$ for notation simplicity.

The $d$-level wavelet transform can be recursively decomposed as~\cite[Appendix~A]{BSB+22}
\begin{equation}
\label{eq:mQWT_recursive}
    W^{(d)}_n = \left(W^{(d-1)}_{n-1} \oplus \id_{n-1}\right) W_n.
\end{equation}
This decomposition follows from the notion of multi-level wavelet transform:
at each level, the transformation is only applied on the low-frequency component (i.e., the top part) of the column vector it acts on.
The wavelet packet transform, however, acts on both the low- and high-frequency components, so we have the decomposition
\begin{equation}
\label{eq:pQWT_recursive}
    P^{(d)}_n = \left(P^{(d-1)}_{n-1}
    \oplus P^{(d-1)}_{n-1}\right) W_n
    = \left(\id_1\otimes P^{(d-1)}_{n-1}\right)W_n
\end{equation}
for the wavelet packet transform.
Equation~\eqref{eq:mQWT_recursive} yields the decomposition 
\begin{equation}
\label{eq:mQWT_decomp}
    W^{(d)}_n =
    \Lambda^{d-1}_0(W_{n-d+1})
    \cdots
    \Lambda^2_0(W_{n-2})
    \Lambda^1_0(W_{n-1})
    W_n,
\end{equation}
where
\begin{equation}
    \Lambda^s_0(W_{n-s}) :=
    \ketbra{0^s}{0^s}\otimes W_{n-s} +
    (\id_s-\ketbra{0^s}{0^s}) \otimes \id_{n-s}
\end{equation}
is the $\ket{0^s}$-controlled unitary operation, for any $s\in\{1,\ldots,d-1\}$.
Similarity, Eq.~\eqref{eq:pQWT_recursive} yields the decomposition
\begin{equation}
    P^{(d)}_n =
    (\id_{d-1}\otimes W_{n-d+1})
    \cdots
    (\id_2\otimes W_{n-2})
    (\id_1\otimes W_{n-1})
    W_n
\end{equation}
for the $d$-level wavelet packet transform.
These decompositions give a simple procedure for implementing a multi-level and packet QWT shown by the quantum circuits in FIG.~\ref{fig:dQWT}

The multi-level packet QWT is construed from single-level QWTs that can be implemented by the method described in~\S\ref{sec:1QWT}.
In contrast, the multi-level QWT is constructed from multi-controlled single-level QWTs.
As in FIG.~\ref{fig:dQWT}(c), we break down these multi-controlled operations in terms of multi-bit Toffoli gates and controlled single-level QWTs.
We discuss an implementation of a multi-bit Toffoli gate in~\S\ref{subsec:keysubs} and a controlled single-level QWT in~\S\ref{subsec:dQWTcost}, where we analyze the complexities of these~operations.

\begin{figure*}
\centering
    \includegraphics[width=.95\linewidth]{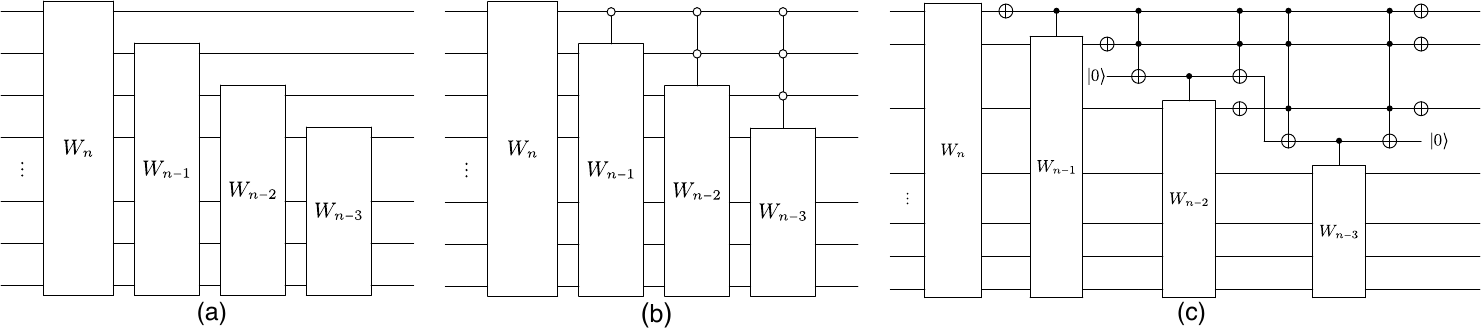}
    \caption{
    Quantum circuits for (a)~the $d$-level packet QWT and~(b) the $d$-level QWT using single-level QWTs; $d=4$ is illustrated here.
    (c) An implementation of multi-controlled single-level QWTs needed for the multi-level QWT in (b) using multi-bit Toffoli gates, controlled single-level QWT and one ancilla qubit that starts and ends in the $\ket{0}$ state. 
    }
    \label{fig:dQWT}
\end{figure*}

\section{Complexity analysis}
\label{sec:complexity}

In this section, we analyze the computational cost of executing single-level, multi-level and packet QWTs, thereby establishing Theorems~\ref{theorem:1QWT}--\ref{theorem:pQWT}.
We begin by analyzing the computational cost of key subroutines in our algorithms in \S\ref{subsec:keysubs}.
We then build upon them and provide cost analysis for the single-level QWT in \S\ref{subsec:sQWTcost} and for the multi-level and packet QWTs in \S\ref{subsec:dQWTcost}.

In our cost analysis and in implementing the key operations, we use ancilla and ``borrowed" qubits.
In contrast to an ancilla qubit that starts from $\ket{0}$ and returns to $\ket{0}$, a borrowed qubit can start from any state and will return to its original state.
The purpose of using borrowed qubits is that they enable simple implementation for complex multi-qubit operations.
The availability of a sufficient number of qubits in our algorithm on which the key operations do not act on them allows us to use them as borrowed qubits in implementing such operations.

\subsection{Complexity of key subroutines}
\label{subsec:keysubs}

Here we analyze the cost of key subroutines used in our algorithm for a single-level QWT: \textsc{prep}, \textsc{select} and \textsc{ushift}, the latter of which adds a classically known constant value to the value encoded in a quantum register.
We also analyze the cost of implementing a multi-qubit reflection, an operation used in the amplitude amplification part of our algorithm.

For simplicity of cost analysis, we state the cost of each key subroutine in a lemma and proceed with analyzing the cost in the poof.
We begin with a lemma stating the cost of executing a multi-bit Toffoli gate,
an operation frequently used in our algorithm and provides an implementation for the multi-qubit reflection about the all-zero state.

\begin{lemma}
\label{lemma:multiToffoli}
    The $(m+1)$-bit Toffoli gate with $m\ge3$, defined as $\Lambda_1^m(X):=\ketbra{1^m}{1^m}\otimes X + (\mathbbm{1}_m-\ketbra{1^m}{1^m})\otimes\mathbbm{1}_1$,
    can be implemented by either of the following computational resources:
    \begin{enumerate}[(I)]
        \item $m-2$ borrowed qubits and $\mathcal{O}(m)$ \textup{\textsc{toffoli}} gates, or 
        \item one borrowed qubit and $\mathcal{O}(m)$ \textup{\textsc{toffoli}} and elementary one- or two-qubit gate.
    \end{enumerate}
\end{lemma}

The implementation based on $m-2$ borrowed qubits follows from Gidney's method~\cite{gid15} for implementing a multi-bit Toffoli gate and the one using one borrowed qubit follows by the method given in Ref.~\cite[Corollary~7.4 ]{BBC+95} and also in Ref.~\cite{HLZ+17}.
Notice that the gate cost of the two methods scales similarly, but one uses only a single borrowed qubit.
However, we sometimes use the method with $m-2$ borrowed qubits due to its simplicity in implementing a multi-bit Toffoli and the availability of a sufficient number of qubits in our algorithm that can be borrowed.

We proceed with the cost of \textsc{select} in the following lemma.

\begin{lemma}
\label{lemma:select}
    \textup{\textsc{select}} in Eq.~\eqref{eq:select} can be executed using one ancilla and one borrowed qubit, two Hadamard and $\mathcal{O}(n)$ \textup{\textsc{not}}, \textup{\textsc{cnot}} and \textup{\textsc{toffoli}} gates.
\end{lemma}
\begin{proof}
    By FIG.~\ref{fig:1QWT}, \textsc{select} is composed of one controlled-$Z$ gate, three \textsc{cnot} gates,
    one controlled-\textsc{sub}, one controlled-\textsc{add} and $n-1$ \textsc{swap} gates.
    The controlled-$Z$ gate can be executed using two Hadamard gates and one \textsc{cnot}, and each \textsc{swap} can be executed using three \textsc{cnot}s.
    By the compilation given in Ref.~\cite{CDK+04}, the \textsc{add} itself can be implemented using one ancilla qubit and $\mathcal{O}(n)$ \textsc{not}, \textsc{cnot} and \textsc{toffoli} gates.
    Hence the controlled-\textsc{add} can be compiled using $\mathcal{O}(n)$ \textsc{cnot}, \textsc{toffoli} and four-bit Toffoli gates,
    the latter of which can be implemented using one borrowed qubit and four \textsc{toffoli} gates by Lemma~\ref{lemma:multiToffoli}.
\end{proof}

In the next lemma, we show that the $m$-qubit reflection $R_m$ about the all-zero state $\ket{0^m}$ can be implemented using $m-2$ borrowed qubits.

\begin{lemma}
\label{lemma:reflect}
    The $m$-qubit reflection $R_m:=\ketbra{0^m}{0^m}-\id_m$ can be executed using
    one ancilla and $m-2$ borrowed qubits along with
    two Hadamard, $2m+2$ \textup{\textsc{not}} and $\mathcal{O}(m)$ \textup{\textsc{toffoli}} gates.
\end{lemma}
\begin{proof}
Using the phase kickback trick and one ancilla qubit, we can implement $R_m$ up to an irrelevant global $-1$ phase factor as
\begin{equation}
\label{eq:Rm}
    (R_m\otimes \mathbbm{1}_1)\ket{\psi}\ket{0}
    = -X^{\otimes m+1}(\id_m\otimes H)\Lambda_1^m(X)
    (\id_m\otimes H) X^{\otimes m+1}\ket{\psi}\ket{0},
\end{equation} 
where $\ket{\psi}$ is any $m$-qubit state and $\Lambda_1^m(X)$ is the $(m+1)$-bit Toffoli gate.
The lemma then follows by Gidney's method~\cite{gid15} for implementing the $(m+1)$-bit Toffoli using $\mathcal{O}(m)$ \textsc{toffoli} gates and $m-2$ borrowed qubits.
\end{proof}

We remark that the $(m+1)$-bit Toffoli can be implemented using only one borrowed qubit and $\mathcal{O}(m)$ \textsc{toffoli} and elementary one- or two-qubit gate by Lemma~\ref{lemma:multiToffoli}.
However, we use the method with $m-2$ borrowed qubits due to its simplicity in implementing a multi-bit Toffoli and the availability of a sufficient number of qubits in our algorithm that can be borrowed.

The following lemma states the cost of adding a known classical value to a quantum register.
We use a controlled version of this operation in our algorithm, the cost of which is stated in the following corollary.
\begin{lemma}
\label{lemma:add}
    Adding a classically known $m$-bit constant to an $n$-qubit register with $m<n$ can be achieved using $m+1$ ancilla qubits and $\mathcal{O}(m)$ \textup{\textsc{not}}, \textup{\textsc{cnot}} and \textup{\textsc{toffoli}} gates. 
\end{lemma}
\begin{proof}
    First, prepare $m$ ancillae in the computational state that encodes the $m$-bit constant.
    This preparation can be achieved by applying at most $m$ \textsc{not} gates.
    Then add this state to the state of the $n$-qubit register by \textsc{add} operation in Eq.~\eqref{eq:add}.
    By $m<n$ and the compilation given in Ref.~\cite{CDK+04}, \textsc{add} can be implemented by one ancilla qubit and $\mathcal{O}(m)$ \textsc{not}, \textsc{cnot} and \textsc{toffoli} gates. 
\end{proof}

The computational cost reported in Lemma~\ref{lemma:add} is indeed the cost of executing \textsc{ushift} in Eq.~\eqref{eq:ushift}.
We use a controlled version of this operation as in the circuit shown in FIG.~\ref{fig:1QWT}.
Because of the \textsc{toffoli} gate in Lemma~\ref{lemma:add}, the controlled-\textsc{ushift} requires implementing a four-bit Toffoli gate, an operation that can be implemented using one borrowed qubit and four \textsc{toffoli} gates by Lemma~\ref{lemma:multiToffoli}.
Therefore, we have the following cost for the controlled \textsc{ushift} as a corollary of Lemma~\ref{lemma:add} and Lemma~\ref{lemma:multiToffoli}.

\begin{corollary}
\label{cor:cushift}
    The controlled-\textup{\textsc{ushift}} operation can be executed by one borrowed qubit, $m+1$ ancilla qubits, and $\mathcal{O}(m)$ \textup{\textsc{cnot}} and \textup{\textsc{toffoli}} gates.
\end{corollary}

The final lemma states the cost of the \textsc{prep} operation.
We remark that the cost of this operation is independent of $n$ as \textsc{prep} generates a quantum state on a number of ancilla qubits that depends on the wavelet order~$M$. 
\begin{lemma}
\label{lemma:prep}
    \textup{\textsc{linprep}} in Eq.~\eqref{eq:linprep} can be executed using $\mathcal{O}(M\log_2M)$ elementary gates and $\ceil{\log_2M}$ borrowed qubits.
    \textup{\textsc{prep}} and \textup{\textsc{unprep}}
    in Eq.~\eqref{eq:prep} can be executed using $\mathcal{O}(M^{3/2})$ elementary gates and one ancilla qubit.
\end{lemma}
\begin{proof}
  The \textsc{linprep} can be implemented using 
  $\mathcal{O}(M)$ rotation gates and $\mathcal{O}(M)$ increment and inverse of increment gates by the procedure given in \S\ref{subsec:prep}.
  The increment gate on $m=\ceil{\log_2M}$ qubits can be implemented using $m$ borrowed qubits and $\mathcal{O}(m)$ elementary gates~\cite{gid_increment}, so the overall gate cost of \textsc{linprep} is $\mathcal{O}(M\log_2M)$.
  As per \S\ref{subsec:prep}, \textsc{prep} and \textsc{unprep} operations can be implemented by preparing the uniform superposition state on $m$ qubits, applying the uniformly controlled rotation in Eq.~\eqref{eq:ucr}, and $\mathcal{O}(\sqrt{M})$ rounds of amplitude amplification.
  The uniform superposition state is prepared by $m$ Hadamard gates, and the uniformly controlled rotation can be implemented by $\mathcal{O}(M)$ \textsc{cnot} and rotation gates~\cite{MVB+04}.
  Therefore, the overall gate cost of \textsc{prep} and \textsc{unprep} is $\mathcal{O}(M^{3/2})$.
\end{proof}

\subsection{Complexity of single-level QWT}
\label{subsec:sQWTcost}

We now build upon the computational cost of the key subroutines analyzed in the previous section to obtain the computational cost of executing a single-level QWT.
To this end, we mainly use Eq.~\eqref{eq:AAoprator} and Eq.~\eqref{eq:OAA}.
By these equations, a single-level QWT is achieved by performing three rotation gates and
\begin{itemize}
    \item Two \textsc{pqwt} and one $\textsc{pqwt}^\dagger$, which by Eq.~\eqref{eq:pqwt} needs performing two \textsc{select} and one $\textsc{select}^\dagger$;
    two \textsc{prep} and one $\textsc{prep}^\dagger$;
    two \textsc{unprep} and one $\textsc{unprep}^\dagger$;
    and one controlled-\textsc{ushift};
    \item Two $(m+1)$-qubit reflection $R_{m+1}$.
\end{itemize}
Therefore, by Lemmas~\ref{lemma:select},\,\,\ref{lemma:reflect},\,\ref{lemma:prep} and Corollary~\ref{cor:cushift} , the gate cost $\mathcal{G}({1\textsc{qwt}})$ for executing a single-level QWT is
\begin{equation}
    \mathcal{G}({1\textsc{qwt}}) =
    3\mathcal{G}(\textsc{select}) +
    6\mathcal{G}(\textsc{prep}) +
    \mathcal{G}(\text{controlled-}\textsc{ushift})+
    2\mathcal{G}(R_{m+1}) + 3
    \in
    \mathcal{O}(n) + 
    \mathcal{O}(M^{3/2})
\end{equation}
where $m=\ceil{\log_2M}$ in our application; $M$ is the wavelet order.
The number of ancilla qubits used is $m+1$:
$m$ ancillae are used for the state-preparation step, and one extra ancilla is the parity qubit \texttt{par}, which is also used in the amplitude amplification~step.

We remark that the borrowed qubits in executing \textsc{prep}, \textsc{select}, controlled-\textsc{ushift} and reflection operations, in Lemmas~\ref{lemma:select}--\ref{lemma:prep}, are borrowed from the portion of quantum registers that these operations do not act on them.
For instance, the $m-2$ borrowed qubits in Lemma~\ref{lemma:reflect} for executing the $m$-qubit reflection $R_m$ could be any $m-2$ qubits of the $n$ qubit register that $R_m$ does not act on them. 
For \textsc{select}, the borrowed qubit is needed to implement the four-bit Toffoli gate, see proof of Lemma.~\ref{lemma:select}, and this qubit could be any qubit in the circuit that the four-bit Toffoli gate does not act on it.
We also remark that the $m+1$ ancilla qubits in Corollary~\ref{cor:cushift} needed for controlled-\textsc{ushift} are qubits of the single-qubit \texttt{par} register and $m$-qubit \texttt{anc} register.
This operation is executed after the amplitude amplification, see FIG.~\ref{fig:1QWT}, when \texttt{par} and \texttt{anc} are in the all-zero state.

Putting all together, the overall gate cost for implementing the single-level QWT is
$\mathcal{O}(n) + \mathcal{O}(M^{3/2})$ and the number of ancilla qubits is $\ceil{\log_2M}+1$.
This is the computational cost reported in Theorem~\ref{theorem:1QWT}.

\subsection{Complexity of multi-level and packet QWTs}
\label{subsec:dQWTcost}

Here we analyze the complexity of implementing a $d$-level and packet QWTs, thereby establishing Theorem~\ref{theorem:dQWT} and Theorem~\ref{theorem:pQWT}.
By FIG.~\ref{fig:dQWT}, implementing a multi-level QWT is achieved by implementing multiply-controlled single-level QWTs.
Our strategy is to break down each multiply-controlled unitaries in terms of multi-bit Toffoli gates and single-controlled unitary.
We then use a compilation for a controlled single-level QWT and an ancilla-friendly compilation for multi-bit Toffoli gates to achieve an efficient yet ancilla-friendly implementation for a multi-level QWT.
The packet QWT, however, is achieved by a sequence of single-level QWTs without controlled qubits, as shown in FIG.~\ref{fig:dQWT}.

Before describing the specifics of our implementation strategy, we first state the complexity of the $\ket{1}$-controlled single-level QWT in the following lemma.
We then build upon this complexity to establish the complexity of multi-level QWT.
\begin{lemma}
\label{lemma:c1QWT}
    The controlled single-level QWT on $n$ qubits, associated with a wavelet of order $M$, can be achieved using $\ceil{\log_2M}+2$ ancilla qubits and $\mathcal{O}(n)+\mathcal{O}(M^{3/2})$ elementary gates.
\end{lemma}
\begin{proof}
    By the circuit in FIG.~\ref{fig:1QWT}, a controlled single-level QWT needs preforming double-controlled-\textsc{sub}, -\textsc{add} and -$\textsc{ushift}$ operations, and single-controlled \textsc{prep} and \textsc{unprep} operations.
    Each \textsc{cnot} is transformed to a \textsc{toffoli}, each \textsc{swap} is transformed to three \textsc{toffoli} gates and~$R$ is transformed to controlled-$R$.
    A double-controlled operation can be reduced to a single-controlled operation using two \textsc{toffoli} gates and one ancilla qubit.
    By the discussion in the proof of Lemma~\ref{lemma:select}, the controlled-\textsc{add} (-\textsc{sub}) can be compiled using $\mathcal{O}(n)$ \textsc{cnot}, \textsc{toffoli} gates.
    The other ancilla qubits are the $m$ qubits used for state preparation and the parity qubit.
    Altogether with Corollary~\ref{cor:cushift} prove the lemma.
\end{proof}
We now proceed with the complexity of $d$-level QWT.
Let the integer $s$, with $1\leq s\leq d$, represent the level of a QWT.
Then for the level $s=r+1$ we need to implement $\ket{0^r}$-controlled-$W_{n-r}$, where $W_{n-r}$ is the single-level QWT on $n-r$ qubits.
For simplicity of cost analysis, we map all $\ket{0}$-controlled operations in FIG.~\ref{fig:dQWT}(b) to $\ket{1}$-controlled operations;
this can be achieved by $2(d-1)$ \textsc{not} gates for $d$-level QWT as in FIG.~\ref{fig:dQWT}(c).
For $r\geq 2$, we implement $\ket{1^r}$-controlled-$W_{n-r}$ by a single ancilla qubit, two $(r+1)$-bit Toffoli gates and one controlled-$W_{n-r}$ as shown in FIG.~\ref{fig:dQWT}(c).
Notice that $s=1$ corresponds to a single-level QWT on $n$ qubits and $s=2$ corresponds to a controlled single-level QWT on $n-1$ qubits.

The gate cost for the controlled single-level QWT on $n-r$ qubits is $\mathcal{O}(n-r)$ by Lemma~\ref{lemma:c1QWT}, disregarding the cost with respect to $M$, and the gate cost for the $(r+1)$-bit Toffoli gate is $\mathcal{O}(r)$ by Lemma~\ref{lemma:multiToffoli}.
Hence the gate cost for each level, including the first and second levels, is $\mathcal{O}(n)$.
We also have an additional gate cost of $\mathcal{O}(M^{3/2})$ for each level associated with the cost of implementing \textsc{prep} and \textsc{unprep}.
We remark that only a single ancilla qubit is used for all levels;
the ancilla qubit starts and ends in $\ket{0}$ for each level to be reused in the next level, as illustrated in FIG.~\ref{fig:dQWT}(c).
Putting all together, we arrive at the computational cost stated in Theorem~\ref{theorem:dQWT} for a $d$-level QWT.

Because the packet QWT does not have multi-controlled operations (see FIG.~\ref{fig:dQWT}(a)), its gate cost simply follows from the cost of the single-level QWT.
The single-level QWT acts on $n-r$ qubits at level $s=r+1$ and has the gate cost $\mathcal{O}(n-r)$ by Theorem~\ref{theorem:1QWT}.
The gate cost for all levels $1\leq s\leq d$ is therefore $\mathcal{O}(dn-d(d-1)/2)$.
We also have an additional gate cost of $\mathcal{O}(M^{3/2})$ for each level associated with the cost of implementing \textsc{prep} and $\textsc{unprep}$, yielding the overall gate cost stated in Theorem~\ref{theorem:pQWT}.
We note that the packet QWT does not need the extra ancilla qubit used in multi-level QWT for implementing the multi-controlled operations.

\section{Discussion and conclusion}
\label{sec:conclusion}

Wavelets and their associated transforms have been extensively used in classical computing. The basis functions of wavelet transforms have features that make such transforms advantageous for numerous applications over their established counterpart, the Fourier transform. However, prior works on developing a quantum analog for wavelet transforms were limited to a few representative cases.
In this paper, we presented quantum algorithms for executing any wavelet transform and a generalized version, the wavelet packet transform, on a quantum computer; the algorithms work for any wavelet of any order. We have established the computational complexity of our algorithms in terms of three parameters involved in wavelet transforms: the wavelet order~$M$, the level~$d$ of wavelet transform, and the number of qubits~$n=\log_2N$ the QWT acts on, with $N$ the dimension of the kernel matrix associated with the wavelet transform.

The core idea of our approach is to express the kernel matrix as a linear combination of~$M$ unitary operations that are simple to implement on a quantum computer and use the LCU technique to construct a probabilistic procedure for implementing the desired QWT. We then make the implementation deterministic using the known success probability of the probabilistic procedure by only a few (two or three) rounds of amplitude amplification.
The gate cost of our algorithm for single-level QWT scales optimally with~$n$, the number of qubits, for the case that the wavelet order~$M$ is constant. Indeed, the order parameter used in practical applications is constant, typically in the range of $2\leq M\leq 20$~\cite{urb08,Bey92,BCR91}.
We also demonstrated that the wavelet filter coefficients become negligibly small for larger values of the wavelet order, making the cost of our algorithms effectively independent of~$M$ for practical applications.
In contrast, the transformation level~$d$ scales linearly with the number of qubits, or~$\log_2N$, for practical applications. Because the value of~$d$ is upper-bounded by~$n$, the gate cost of multi-level and packet QWTs scales as~$\mathcal{O}(n^2)$ in the worst case.
Even for the worst case, our algorithm improves the gate cost of prior works on the second- and fourth-order Daubechies QWT from $\mathcal{O}(n^3)$ to $\mathcal{O}(n^2)$. 

We remark that our approach requires a number of ancilla qubits that scales as $\log_2M$ with the wavelet order.
The number of ancilla qubits would be a small constant number considering the range of wavelet order or the magnitude of wavelet coefficients in practical applications.
A potential area for further exploration is constructing ancilla-free quantum algorithms for all QWTs.
Constructing such algorithms would be valuable for early fault-tolerant quantum computers with limited qubits and is plausible because QWTs are unitary transformations.
More importantly, a primary area for future research is exploring the opportunities offered by quantum wavelet transforms in quantum algorithms, particularly in simulating quantum systems~\cite{BRS+15,BSB+22,BNW+23} and image processing~\cite{ME22,TZZ+23,ZHH+24} where wavelet transforms could be advantageous over the established Fourier transform.

\section*{Acknowledgements}
We acknowledge the generous support and funding of this project by the Defense Advanced Research Projects Agency (DARPA) under the grant number HR0011-23-3-0021.
We also acknowledge support from the Canada 150 Research Chairs program and NSERC-IRC. A.A.-G. also acknowledges the generous support of Anders G. Frøseth.


\bibliography{references}
\end{document}